\documentclass[%
 reprint,
superscriptaddress,
 amsmath,amssymb,
 aps,
pra,
showkeys
]{revtex4-2}

\usepackage{graphicx}% Include figure files
\usepackage{dcolumn}% Align table columns on decimal point
\usepackage{bm}% bold math
\usepackage{algorithm}%
\usepackage{algorithmicx}%
\usepackage[noend]{algpseudocode}%

\usepackage{amsthm}
\usepackage{amssymb}
\usepackage{amsmath}
\usepackage{bm}
\usepackage{graphicx}
\usepackage{subfigure}
\usepackage{ragged2e} 
\usepackage{siunitx}
\usepackage{footmisc}
\usepackage{tikz}
\usepackage{amsfonts}
\usepackage{braket}
\usepackage{hyperref}
\usepackage[table]{xcolor}

\newtheorem{theorem}{Theorem}%  meant for continuous 
\newtheorem{corollary}{Corollary}%
\newtheorem{claim}{Claim}
\newtheorem{definition}{Definition}
\renewenvironment{proof}{\paragraph*{Proof:}}{\hfill$\square$}
\allowdisplaybreaks

\begin{document}

\preprint{APS/123-QED}

\title{Expressivity Limits in Quantum Walk-based Optimization}% Force line breaks with \\

\author{Guilherme Adamatti Bridi}
\email{gabridi@cos.ufrj.br}
\affiliation{%
 Federal University of Rio de Janeiro, Brazil
}

\author{Debbie Lim}
\email{limhueychih@gmail.com}
\affiliation{%
 Centre for Quantum Computer Science, Faculty of Sciences and Technology, University of Latvia, Latvia
}

\author{Lirandë Pira}%
\email{lpira@nus.edu.sg}
\affiliation{%
 Centre for Quantum Technologies, National University of Singapore, Singapore
}

\author{Raqueline Azevedo Medeiros Santos}
\email{rsantos@lu.lv}
\affiliation{%
 Centre for Quantum Computer Science, Faculty of Sciences and Technology, University of Latvia, Latvia
}

\author{Franklin de Lima Marquezino}%
\email{franklin@cos.ufrj.br}
\affiliation{%
 Federal University of Rio de Janeiro, Brazil
}

\author{Soumik Adhikary}%
\email{soumik@nus.edu.sg}
\affiliation{%
 Centre for Quantum Technologies, National University of Singapore, Singapore
}

\date{\today}% It is always \today, today,
             %  but any date may be explicitly specified

\begin{abstract}
Quantum algorithms have emerged as a promising tool to solve combinatorial optimization problems. The quantum walk optimization algorithm (QWOA) is one such variational approach that has recently gained attention. In the broader context of variational quantum algorithms (VQAs), understanding the expressivity of the ansatz has proven critical for evaluating their performance. A key method to study this aspect involves analyzing the dimension of the dynamic Lie algebra (DLA). In this work, we derive novel upper bounds on the DLA dimension for QWOA applied to arbitrary optimization problems. Specifically, we show that the DLA dimension scales at most quadratically with the number of distinct eigenvalues of the problem Hamiltonian. As a consequence, our bound guarantees a polynomial DLA dimension with respect to the input size for optimization problems in the class $\mathsf{NPO}\text{-}\mathsf{PB}$. This result, coupled with recently established performance bounds for QWOA, allows us to identify complexity-theoretic conditions under which QWOA must be overparameterized to obtain optimal or approximate solutions for $\mathsf{NPO}\text{-}\mathsf{PB}$ problems.
\end{abstract}

%\keywords{}%Use showkeys class option if keyword
                              %display desired
\maketitle

%\tableofcontents

\section{Introduction}
Combinatorial optimization problems play an important role in modeling and solving real-world systems. Some common classical approaches to solving these problems include Lagrange multipliers, simulated annealing, branch-and-bound, and evolutionary algorithms~\cite{comb_opt_1,comb_opt_2}. Recent advances in quantum computing have also led to the development of quantum algorithms designed to solve these problems~\cite{abbas2024challenges, perez2024variational, gemeinhardt2023quantum}. Among them, variational quantum algorithms (VQAs)~\cite{cerezo2021variational} --- which combine classical optimization with parametrized quantum circuits (PQCs), i.e., quantum circuits composed of unitary operations controlled by tunable parameters --- have emerged as a prominent class and are particularly well suited for noisy quantum devices. Arguably, the most prominent VQA for combinatorial optimization is the quantum approximate optimization algorithm (QAOA)~\cite{farhi2014quantum}, a gate-based heuristic inspired by the quantum adiabatic algorithm~\cite{qaa1,qaa2} that alternates between problem-specific and mixing unitaries to explore the solution space. 

Several variants of QAOA have been proposed in the literature. Many of these approaches modify the mixing unitary in order to incorporate structural information or to better respect problem constraints. Examples of structure-aware designs include ADAPT-QAOA~\cite{zhu2022adaptive} which builds the mixer iteratively based on problem-informed operator selection and adaptive-bias QAOA (ab-QAOA)~\cite{yu2022quantum} that introduces local bias fields within the standard mixer. On the other hand, in the front of constraint-aware designs, the quantum walk optimization algorithm (QWOA)~\cite{marsh2019quantum,marsh2020combinatorial} is a special variant of QAOA where the standard QAOA mixing unitary is replaced with a continuous-time quantum walk (CTQW)~\cite{farhi1998quantum} operator. QWOA has been proposed as a natural alternative for solving combinatorial optimization problems with structured constraints, where quantum walks can be efficiently implemented through an indexing procedure that maps each feasible solution to a unique computational basis state.

One of the key aspects for the success of a VQA is the expressivity of the PQC (a.k.a. ansatz), which should be sufficient to prepare the state that encodes the solution to the considered problem. Several approaches have been proposed to study the expressivity of PQCs, including techniques based on covering numbers~\cite{sain1996nature, du2022efficient} and $t$-designs~\cite{harrow2009random}. Yet another mathematical tool for measuring expressivity, which we would focus on in this work, is the dynamic Lie algebra (DLA)~\cite{larocca2022diagnosing, larocca2023theory, meyer2023exploiting, ragone2024lie, allcock2024dynamical}. Given the set of Hamiltonians that generate the unitaries in a PQC, the corresponding DLA is constructed by taking repeated nested commutators of the Hamiltonians. This forms a Lie algebra that characterizes the space of unitaries accessible through time evolution. As noted by Larocca \emph{et al.}~\cite{larocca2023theory}, a PQC is considered overparametrized when the number of parameters exceeds the DLA dimension. In this regime, one can explore all independent directions in the state space, as allowed by the Lie algebra, by tuning the parameters of the ansatz.

Beyond expressivity, the dimension of the DLA also influences other key aspects of VQAs, such as trainability and the classical simulability. In the vein of classical simulability, Goh \emph{et al.}~\cite{goh2023lie} showed that, given some additional conditions, if the DLA dimension is polynomial in the input size, one can efficiently compute expectation values of observables belonging to this algebra; this is the case for both QAOA and QWOA. Regarding trainability, the dimension of the DLA is closely related to the emergence of barren plateaus~\cite{mcclean2018barren} --- regions in the parameter space where the variance of the loss function and its gradients vanishes exponentially with system size. In such regions, the optimization landscape becomes effectively flat, naturally affecting gradient-based optimization methods and, notably, also limiting the effectiveness of gradient-free approaches~\cite{arrasmith2021effect}, thus challenging the training of PQCs. In particular, Ragone \emph{et al.}~\cite{ragone2024lie} established that the dimension of the DLA is inversely proportional to the variance of the loss function. This provides a rigorous framework for diagnosing the presence of barren plateaus based on algebraic properties of the circuit. 

In the context of QAOA, DLA has been studied both numerically and analytically~\cite{larocca2022diagnosing, larocca2023theory, meyer2023exploiting, allcock2024dynamical} for the Max-Cut problem, a combinatorial optimization problem defined on graphs and widely studied in the QAOA literature~\cite{blekos2024review}. On the numerical side, the DLA dimension has been evaluated for two specific classes: $2$-regular graphs (also known as cycle graphs or ring of disagrees) and chain graphs (or path graphs), achieving scaling of $\mathcal{O}(n)$ and $\mathcal{O}(n^2)$, respectively, where $n$ is the number of vertices of the graph~\cite{larocca2022diagnosing, meyer2023exploiting}. On the analytical side, Allcock \emph{et al.}~\cite{allcock2024dynamical}, using tools from group theory, proved that the exact dimension of the DLA for $2$-regular graphs is indeed linear in $n$, and also established that the exact DLA dimension for the complete graph is in the order of $n^3$. 

Together, these numerical and analytical findings have implications in the context of expressivity of the QAOA ansatz for the Max-Cut problem. For instance, numerical findings by Larocca \emph{et al.}~\cite{larocca2023theory} showed that the number of layers required to solve Max-Cut with a high success probability, on $2$-regular graphs and chains, matches the scaling of their corresponding DLA dimensions \footnote{Note that the number of variational parameters of QAOA/QWOA is linearly related to its number of layers. So, in asymptotic discussions, we can use the number of layers of QAOA/QWOA as a metric for parameterization without loss of generality.}. Moreover, in the case of Max-Cut on $2$-regular graphs, this correspondence was later established analytically in Ref.~\cite{rabinovich2025role}. This allows us to infer that, for these classes of graphs, QAOA does not require overparameterization to achieve optimal solutions.

While DLA has been extensively studied in the context of QAOA, it remains underexplored in the setting of QWOA. In this work, we initiate a systematic study of DLAs within the QWOA framework, focusing on whether overparameterization is necessary for solving or approximating combinatorial optimization problems.

\paragraph{\bf Main Contributions.} 
In this work, we establish a novel bound on the dimension of DLA for QWOA with arbitrary combinatorial optimization problems, specifically showing that it scales at most quadratically with the number of distinct eigenvalues of the problem Hamiltonian (Theorem~\ref{thm:DLA_QWOA}). Our bound is intuitive given that QWOA is fundamentally independent of the underlying structure of optimization problems, meaning that the unitary dynamics of QWOA depends only on the probability distribution associated with the problem Hamiltonian spectrum~\cite{headley2023problem}, whose size equals the number of distinct eigenvalues. Furthermore, a notable feature of our approach lies in its simplicity, as it avoids traditional DLA analyses based on Pauli decompositions or advanced group-theoretic tools by relying instead on elementary spectral arguments enabled by the structural properties of QWOA. 

Our general bound directly implies that, for $\mathsf{NP}$ optimization problems with polynomially bounded cost functions ($\mathsf{NPO}\text{-}\mathsf{PB}$ problems) --- the class of problems for which QWOA was designed~\cite{marsh2019quantum,marsh2020combinatorial,bennett2021quantum1} --- the dimension of the associated DLA is polynomially bounded in the input size (Corollary~\ref{cor:NPO-BP-DLA}). This result, when combined with new insights (formalized in Theorem~\ref{thm:QWOA_nonpoly}) into the analytical results of Refs.~\cite{bridi2024analytical, xie2025performance} concerning the intrinsic limitation of the performance of QWOA, being bounded by a quadratic speed-up over the random sampling procedure --- representing a Grover-style speed-up~\cite{grover1996fast} --- leads to complexity-theoretic conditions for QWOA to be overparameterized (Theorem~\ref{thm:NPOPB}). Specifically, we argue that for any optimization problem belonging to the complexity class $\mathsf{NPO}\text{-}\mathsf{PB}$ but not to the complexity class $\mathsf{BPPO}$ (see Definition~\ref{def:BPPO}), the QWOA ansatz must be overparameterized in order to find optimal solutions. Similarly, if the problem does not belong to the class $\mathsf{BP\text{-}APX}$ (see Definition~\ref{def:BPAPX}), then QWOA also requires overparameterization to achieve any fixed approximation ratio.

\paragraph{\bf Outline.} 
This paper is organized as follows. In Section~\ref{sec:preliminaries_and_methods}, we introduce some background on combinatorial optimization, complexity classes, QWOA, and DLA. We present our main results and their theoretical implications in Section~\ref{sec:main_results}. Section~\ref{sec:examples} illustrates these results through the well-known problems of unstructured search, Max-Cut, and $k$-Densest Subgraph. Given the generality of the theoretical results presented, the examples included are not intended to serve as benchmarks, but rather to illustrate broader underlying phenomena. Lastly, we conclude our findings and discuss the future directions of our work in Section~\ref{sec:conclusion}.

\section{Preliminaries and Methods}\label{sec:preliminaries_and_methods}

\subsection{Combinatorial Optimization Problems}

Combinatorial optimization focuses on identifying the optimal solution from a discrete, finite set of candidates. Such problems are prevalent in both theoretical and applied settings, including graph problems, logistics, scheduling, and network design~\cite{comb_opt_1, comb_opt_2, braekers2016vehicle}. Their central difficulty lies in the exponential growth of the solution space, which makes exhaustive search impractical even for moderately sized instances~\cite{crescenzi1995compendium}. Quantum algorithms based on Grover's search~\cite{grover1996fast}, such as Grover Adaptive Search~\cite{gas1, gas2, gas3}, are also affected by this limitation, since their quadratic speed-up is insufficient to overcome the exponential size of the solution space in general combinatorial problems. In response to these challenges, a rich literature has developed over the past decades exploring algorithmic strategies for combinatorial optimization problems, ranging from exact
algorithms~\cite{laporte1987exact, baker1983exact, gupta2018faster}, to approximate and heuristics methods~\cite{festa2002randomized, sanchez2020systematic, mazyavkina2021reinforcement}, as well as quantum algorithms~\cite{abbas2024challenges, perez2024variational, gemeinhardt2023quantum}.

We consider combinatorial optimization problems defined by a cost function (the objective function) $C : \mathcal{S} \rightarrow \mathbb{R}$, where $\mathcal{S}$ --- the solution space --- is the set of all possible solutions. The feasible solution space $\mathcal{S'} \subseteq \mathcal{S}$ is the set of all feasible solutions, i.e., those satisfying problem constraints. The goal is to find the feasible solution $z^* \in \mathcal{S'}$ that minimizes the cost function, i.e., $z^* = \arg\min_{z \in \mathcal{S'}} C(z)$. A problem is said to be unconstrained if $\mathcal{S'} = \mathcal{S}$, and constrained otherwise. We denote the input size of a given problem instance as $s$ \footnote{The input size is commonly denoted by $n$ in the literature, often coinciding with the number of vertices in a graph and/or the number of qubits. However, we adopt a different symbol in this work, as we deal with constrained problems where such a correspondence does not apply.}. Let $\mathcal{C}$ be the image of $C$ when restricted to the subdomain $\mathcal S'$. We denote the set  $\mathcal{C} = \{ x_1, x_2, \dots, x_m \}$, where each $x_j$ represents a distinct (feasible) cost value and $m = |\mathcal{C}|$. For $j = 1, \ldots,m$, the respective multiplicity of each cost $x_j$ is denoted by $k_j$ (i.e., $k_j = |\{ s \in \mathcal{S}' \mid C(s) = x_j \}|$).

\subsection{Computational Complexity Classes}

Many important problems in computer science, such as combinatorial optimization problems, have been extensively studied from a complexity perspective. Computational complexity, investigated since the
1970s~\cite{sat, karp, lenstra1979computational, cormen2009book}, is the study of the resources, such as time and space, required by algorithms to solve a given problem, abstracting away from specific machine details to focus on the inherent difficulty of the problem itself. Over time, these early investigations led to the formalization of complexity classes that capture the inherent difficulty of computational problems.

Here, we present the relevant computational complexity classes that would be instrumental for our analytical results. While some of these classes are well-known, others are introduced or adapted to address probabilistic and optimization settings specific to our purposes. We begin by listing standard classes for decision and optimization problems.

\begin{itemize}
    \item The class $\mathsf{P}$ (Polynomial Time)~\cite{papadimitriou2003computational} consists of decision problems that can be solved in polynomial time by a deterministic algorithm. 
    \item The class $\mathsf{NP}$ (Non-deterministic Polynomial Time)~\cite{papadimitriou2003computational} consists of decision problems whose “yes” certificates can be verified in polynomial time by a deterministic algorithm.
    \item The class $\mathsf{BPP}$ (Bounded-Error Probabilistic Polynomial Time)~\cite{papadimitriou2003computational} consists of decision problems that can be solved in polynomial time by a randomized (or possibly deterministic) algorithm with high probability of correctness.
    \item The class $\mathsf{PO}$ ($\mathsf{P}$ Optimization)~\cite{manyem2008syntactic} consists of optimization problems that can be solved in polynomial time by a deterministic algorithm.
    \item The class $\mathsf{NPO}$ (NP Optimization)~\cite{krentel1986complexity} consists of optimization problems for which the solutions have polynomial size, and both their feasibility and cost can be determined in polynomial time by a deterministic algorithm.
    \item The class $\mathsf{NPO}\text{-}\mathsf{PB}$ ($\mathsf{NPO}$ Polynomially Bounded)~\cite{krentel1986complexity} consists of $\mathsf{NPO}$ problems whose cost function takes values over a discrete range that is polynomially bounded in the size of the input instance.
    \item The class $\mathsf{APX}$ (Approximable)~\cite{ausiello2012complexity} consists of optimization problems for which there exists a polynomial time deterministic algorithm that produces solutions with a guaranteed fixed approximation ratio \footnote{The approximation ratio, a widely used metric to measure the performance of non-exact algorithms, is defined for a given problem instance and an algorithm as the ratio between the cost of the solution output by the algorithm and the optimal cost of the problem. For a given problem, algorithms that guarantee a fixed approximation factor for every instance are called approximation algorithms.}. 
\end{itemize}

Now, we define some computational classes that, to the best of our knowledge, have not been explicitly introduced in the literature. Firstly, in analogy to the way the classes $\mathsf{PO}$ and $\mathsf{NPO}$ respectively extend $\mathsf{P}$ and $\mathsf{NP}$ to optimization problems, we propose similar extensions for the class $\mathsf{BPP}$ as follows.

\begin{definition}\label{def:BPPO}
The class $\mathsf{BPPO}$ ($\mathsf{BPP}$ Optimization) consists of optimization problems that can be solved in polynomial time by a randomized (or possibly deterministic) algorithm with high probability.
\end{definition}

Finally, in analogy to the way the classes $\mathsf{BPP}$ and $\mathsf{BPPO}$ introduce randomized algorithms to the classes $\mathsf{P}$ and $\mathsf{PO}$, respectively, we introduce randomized algorithms to the $\mathsf{APX}$ class in the proposed class below.

\begin{definition}\label{def:BPAPX}
The class $\mathsf{BP}\text{-}\mathsf{APX}$ (Bounded-Error Probabilistic $\mathsf{APX}$) consists of optimization problems for which there exists a polynomial time randomized (or possibly deterministic) algorithm that, with high probability, produces solutions with a guaranteed fixed approximation ratio.
\end{definition}

Alternatively, $\mathsf{BP}\text{-}\mathsf{APX}$ introduces approximation algorithms for the class $\mathsf{BPPO}$. More generally, Figure~\ref{fig:complexity_diagram} illustrates the relationship between the majority of the complexity classes described in this section, including the newly defined ones. Some of our results will need to use $\mathsf{BPPO}$ and $\mathsf{BP}\text{-}\mathsf{APX}$ in their assumptions. We note here that either both classes are the same, or $\mathsf{BPPO}$ is slightly larger than $\mathsf{PO}$. The same goes for $\mathsf{BP}\text{-}\mathsf{APX}$ and $\mathsf{APX}$. Moreover, since it is unknown whether $\mathsf{BPPO} \subseteq \mathsf{NPO}$ or whether $\mathsf{NPO} \subseteq \mathsf{BPPO}$, we can only conclude that $\mathsf{PO} \subseteq \mathsf{BPPO} \cap \mathsf{NPO}$.

\begin{figure}[H]
    \centering
            \begin{tikzpicture}[scale=0.5]

          \tikzstyle{main node}=[circle, draw, minimum size=1.2cm, inner sep=0pt]

          \node[main node] (A) at (0, 0) {\scriptsize $\mathsf{P}$};
          \node[main node] (B) at (0, -6) {\scriptsize $\mathsf{BPP}$};
          \node[main node] (C) at (5, 0) {\scriptsize $\mathsf{PO}$};
          \node[main node] (D) at (5, -6) {\scriptsize $\mathsf{BPPO}$};
          \node[main node] (E) at (10, 0) {\scriptsize $\mathsf{APX}$};
          \node[main node] (F) at (10, -6) {\scriptsize $\mathsf{BP}\text{-}\mathsf{APX}$};

         \draw[->, double, ] (A) -- (B);
         \draw[->, double, ] (A) -- (C);
         \draw[->, double, ] (B) -- (D);
         \draw[->, double, ] (C) -- (D);
         \draw[->, double, ] (C) -- (E);
         \draw[->, double, ] (D) -- (F);
         \draw[->, double, ] (E) -- (F);
          
         \draw[dashed] (-2, -3) -- (12, -3);
         \draw[dashed] (2.5, 2) -- (2.5, -8);
         \draw[dashed] (7.5, 2) -- (7.5, -8);
        
         \node at (-3.5, -3) {\small \text{Rand.}};
         \node at (2.5, 3) {\small \text{Opt.}};
         \node at (7.5, 3) {\small \text{Approx.}};
            
            \end{tikzpicture}
    \caption{A diagram relating complexity classes. Dashed lines indicate the separation between deterministic and randomized algorithms (horizontal), decision and optimization problems (first vertical line), and exact and approximation algorithms (second vertical line). Observe that, starting from $\mathsf{P}$, one can reach any other class by progressively introducing, if necessary, optimization, randomized algorithms, and approximation algorithms.}
    \label{fig:complexity_diagram}
\end{figure}

{\bf Asymptotic notation.} In complexity theory, understanding the asymptotic scaling behavior is central to characterizing the resources required by an algorithm or computational task. To this end, we use standard asymptotic notation to describe the growth of functions. Given functions $f(n)$ and $g(n)$, we write $f(n) = \mathcal{O}(g(n))$ to indicate that $f$ grows no faster than $g$ up to a constant factor, i.e., there exist constants $c, n_0 > 0$ such that $f(n) \leq c\,g(n)$ for all $n \geq n_0$. Conversely, $f(n) = \Omega(g(n))$ means that $f$ grows at least as quickly as $g$ asymptotically, i.e, $f(n) \geq c\,g(n)$ for all $n\geq n_0$. The notation $f(n) = \Theta(g(n))$ denotes that $f$ and $g$ grow at the same rate up to constant factors, i.e, $f(n)=\mathcal{O}(g(n))$ and $f(n) = \Omega(g(n))$. We also use the little-omega notation $f(n) = \omega(g(n))$ when $f$ grows strictly faster than $g$, i.e., for all $c > 0$, there is a $n_0 > 0$ such that $f(n) > c\,g(n)$ for all $n \geq n_0$. 

We write $\operatorname{poly}(n)$ to denote the set 
$\operatorname{poly}(n) = \{h(n): h(n) \text{ is a polynomial function of } $n$ \}$. Using this definition, statements such as $f(n) = \mathcal{O}(\text{poly}(n))$ or $f(n) = \omega(\text{poly}(n))$ are shorthand for saying that $f(n)$ grows at most like some polynomial (for $\mathcal{O}(\text{poly}(n))$) or strictly faster than any polynomial (for $\omega(\text{poly}(n))$). Formally, for $\mathcal{O}(\text{poly}(n))$ there exists a constant $k$ such that $f(n) = \mathcal{O}(n^k)$, and for $\omega(\text{poly}(n))$ we have $\omega(n^k)$ for all $k$.

\subsection{Quantum Walk Optimization Algorithm (QWOA)}\label{app:QWOA}
The quantum walk optimization algorithm (QWOA)~\cite{marsh2019quantum,marsh2020combinatorial} is a generalization of the well-known quantum approximate optimization algorithm (QAOA)~\cite{farhi2014quantum}. QAOA was proposed to approximately solve unconstrained combinatorial optimization problems. The algorithm is inspired by a Trotterized discretization of the quantum adiabatic algorithm~\cite{qaa1,qaa2}, and approximates its dynamics in the large-depth limit~\cite{farhi2014quantum}. QWOA extends the QAOA framework for constrained combinatorial optimization problems, where the mixing operator in QAOA is replaced with a continuous-time quantum walk (CTQW)~\cite{farhi1998quantum} over the space of feasible solutions. More specifically, with the aid of efficient indexing and un-indexing operators, we work within an indexed subspace where the CTQW can be efficiently implemented~\cite{marsh2020combinatorial}. Typically, QWOA is implemented using a CTQW on a complete graph~\cite{slate2021quantum,bennett2021quantum1,bennett2021quantum2,matwiejew2024quantum}, which is also the version we consider in this work \footnote{Note that, up to a rescaling of the variational parameter, QWOA is equivalent to a variant of QAOA called the Grover Mixer Quantum Alternating Operator Ansatz (GM-QAOA)~\cite{bartschi2020grover}. Accordingly, all results and discussions in this work apply directly to that ansatz as well, and vice versa, meaning that results from that ansatz can be directly applied here.}.

\begin{figure*}[ht]
\centering

\subfigure[]{
\begin{minipage}{0.21\textwidth}
\begin{tikzpicture}[>=stealth, scale=0.65]

\node at (-3,4) {$S'$};
\node at (3,4) {$\{0,\dots,9\}$};

\node (s0)  at (-3,3)    {0};
\node (s1)  at (-3,2.3)  {1};
\node (s3)  at (-3,1.6)  {3};
\node (s4)  at (-3,0.9)  {4};
\node (s6)  at (-3,0.2)  {6};
\node (s8)  at (-3,-0.5) {8};
\node (s9)  at (-3,-1.2) {9};
\node (s11) at (-3,-1.9) {11};
\node (s13) at (-3,-2.6) {13};
\node (s15) at (-3,-3.3) {15};

\node (i0) at (3,3)    {0};
\node (i1) at (3,2.3)  {1};
\node (i2) at (3,1.6)  {2};
\node (i3) at (3,0.9)  {3};
\node (i4) at (3,0.2)  {4};
\node (i5) at (3,-0.5) {5};
\node (i6) at (3,-1.2) {6};
\node (i7) at (3,-1.9) {7};
\node (i8) at (3,-2.6) {8};
\node (i9) at (3,-3.3) {9};

\draw[->] (s0)  -- (i7);  
\draw[->] (s1)  -- (i5);  
\draw[->] (s3)  -- (i4); 
\draw[->] (s4)  -- (i9);
\draw[->] (s6)  -- (i1);
\draw[->] (s8)  -- (i2); 
\draw[->] (s9)  -- (i6);
\draw[->] (s11) -- (i0); 
\draw[->] (s13) -- (i8); 
\draw[->] (s15) -- (i3);

\end{tikzpicture}
\end{minipage}
}
\hfill
\subfigure[]{
\begin{minipage}{0.30\textwidth}
\centering
\scriptsize
\[
\left(
\begin{array}{cccccccccccccccc}
\bm{0}&\bm{1}&{\color{gray} 0}&\bm{1}&\bm{1}&{\color{gray} 0}&\bm{1}&{\color{gray} 0}&\bm{1}&\bm{1}&{\color{gray} 0}&\bm{1}&{\color{gray} 0}&\bm{1}&{\color{gray} 0}&\bm{1}\\
\bm{1}&\bm{0}&{\color{gray} 0}&\bm{1}&\bm{1}&{\color{gray} 0}&\bm{1}&{\color{gray} 0}&\bm{1}&\bm{1}&{\color{gray} 0}&\bm{1}&{\color{gray} 0}&\bm{1}&{\color{gray} 0}&\bm{1}\\
{\color{gray} 0}&{\color{gray} 0}&{\color{gray} 0}&{\color{gray} 0}&{\color{gray} 0}&{\color{gray} 0}&{\color{gray} 0}&{\color{gray} 0}&{\color{gray} 0}&{\color{gray} 0}&{\color{gray} 0}&{\color{gray} 0}&{\color{gray} 0}&{\color{gray} 0}&{\color{gray} 0}&{\color{gray} 0}\\
\bm{1}&\bm{1}&{\color{gray} 0}&\bm{0}&\bm{1}&{\color{gray} 0}&\bm{1}&{\color{gray} 0}&\bm{1}&\bm{1}&{\color{gray} 0}&\bm{1}&{\color{gray} 0}&\bm{1}&{\color{gray} 0}&\bm{1}\\
\bm{1}&\bm{1}&{\color{gray} 0}&\bm{1}&\bm{0}&{\color{gray} 0}&\bm{1}&{\color{gray} 0}&\bm{1}&\bm{1}&{\color{gray} 0}&\bm{1}&{\color{gray} 0}&\bm{1}&{\color{gray} 0}&\bm{1}\\
{\color{gray} 0}&{\color{gray} 0}&{\color{gray} 0}&{\color{gray} 0}&{\color{gray} 0}&{\color{gray} 0}&{\color{gray} 0}&{\color{gray} 0}&{\color{gray} 0}&{\color{gray} 0}&{\color{gray} 0}&{\color{gray} 0}&{\color{gray} 0}&{\color{gray} 0}&{\color{gray} 0}&{\color{gray} 0}\\
\bm{1}&\bm{1}&{\color{gray} 0}&\bm{1}&\bm{1}&{\color{gray} 0}&\bm{0}&{\color{gray} 0}&\bm{1}&\bm{1}&{\color{gray} 0}&\bm{1}&{\color{gray} 0}&\bm{1}&{\color{gray} 0}&\bm{1}\\
{\color{gray} 0}&{\color{gray} 0}&{\color{gray} 0}&{\color{gray} 0}&{\color{gray} 0}&{\color{gray} 0}&{\color{gray} 0}&{\color{gray} 0}&{\color{gray} 0}&{\color{gray} 0}&{\color{gray} 0}&{\color{gray} 0}&{\color{gray} 0}&{\color{gray} 0}&{\color{gray} 0}&{\color{gray} 0}\\
\bm{1}&\bm{1}&{\color{gray} 0}&\bm{1}&\bm{1}&{\color{gray} 0}&\bm{1}&{\color{gray} 0}&\bm{0}&\bm{1}&{\color{gray} 0}&\bm{1}&{\color{gray} 0}&\bm{1}&{\color{gray} 0}&\bm{1}\\
\bm{1}&\bm{1}&{\color{gray} 0}&\bm{1}&\bm{1}&{\color{gray} 0}&\bm{1}&{\color{gray} 0}&\bm{1}&\bm{0}&{\color{gray} 0}&\bm{1}&{\color{gray} 0}&\bm{1}&{\color{gray} 0}&\bm{1}\\
{\color{gray} 0}&{\color{gray} 0}&{\color{gray} 0}&{\color{gray} 0}&{\color{gray} 0}&{\color{gray} 0}&{\color{gray} 0}&{\color{gray} 0}&{\color{gray} 0}&{\color{gray} 0}&{\color{gray} 0}&{\color{gray} 0}&{\color{gray} 0}&{\color{gray} 0}&{\color{gray} 0}&{\color{gray} 0}\\
\bm{1}&\bm{1}&{\color{gray} 0}&\bm{1}&\bm{1}&{\color{gray} 0}&\bm{1}&{\color{gray} 0}&\bm{1}&\bm{1}&{\color{gray} 0}&\bm{0}&{\color{gray} 0}&\bm{1}&{\color{gray} 0}&\bm{1}\\
{\color{gray} 0}&{\color{gray} 0}&{\color{gray} 0}&{\color{gray} 0}&{\color{gray} 0}&{\color{gray} 0}&{\color{gray} 0}&{\color{gray} 0}&{\color{gray} 0}&{\color{gray} 0}&{\color{gray} 0}&{\color{gray} 0}&{\color{gray} 0}&{\color{gray} 0}&{\color{gray} 0}&{\color{gray} 0}\\
\bm{1}&\bm{1}&{\color{gray} 0}&\bm{1}&\bm{1}&{\color{gray} 0}&\bm{1}&{\color{gray} 0}&\bm{1}&\bm{1}&{\color{gray} 0}&\bm{1}&\bm{0}&{\color{gray} 0}&{\color{gray} 0}&\bm{1}\\
{\color{gray} 0}&{\color{gray} 0}&{\color{gray} 0}&{\color{gray} 0}&{\color{gray} 0}&{\color{gray} 0}&{\color{gray} 0}&{\color{gray} 0}&{\color{gray} 0}&{\color{gray} 0}&{\color{gray} 0}&{\color{gray} 0}&{\color{gray} 0}&{\color{gray} 0}&{\color{gray} 0}&{\color{gray} 0}\\
\bm{1}&\bm{1}&{\color{gray} 0}&\bm{1}&\bm{1}&{\color{gray} 0}&\bm{1}&{\color{gray} 0}&\bm{1}&\bm{1}&{\color{gray} 0}&\bm{1}&{\color{gray} 0}&\bm{1}&{\color{gray} 0}&\bm{0}
\end{array}
\right)
\]
\end{minipage}
}
\hfill
\subfigure[]{
\begin{minipage}{0.30\textwidth}
\centering
\scriptsize
\[
\left(
\begin{array}{cccccccccccccccc}
\bm{0}&\bm{1}&\bm{1}&\bm{1}&\bm{1}&\bm{1}&\bm{1}&\bm{1}&\bm{1}&\bm{1}&{\color{gray} 0}&{\color{gray} 0}&{\color{gray} 0}&{\color{gray} 0}&{\color{gray} 0}&{\color{gray} 0}\\
\bm{1}&\bm{0}&\bm{1}&\bm{1}&\bm{1}&\bm{1}&\bm{1}&\bm{1}&\bm{1}&\bm{1}&{\color{gray} 0}&{\color{gray} 0}&{\color{gray} 0}&{\color{gray} 0}&{\color{gray} 0}&{\color{gray} 0}\\
\bm{1}&\bm{1}&\bm{0}&\bm{1}&\bm{1}&\bm{1}&\bm{1}&\bm{1}&\bm{1}&\bm{1}&{\color{gray} 0}&{\color{gray} 0}&{\color{gray} 0}&{\color{gray} 0}&{\color{gray} 0}&{\color{gray} 0}\\
\bm{1}&\bm{1}&\bm{1}&\bm{0}&\bm{1}&\bm{1}&\bm{1}&\bm{1}&\bm{1}&\bm{1}&{\color{gray} 0}&{\color{gray} 0}&{\color{gray} 0}&{\color{gray} 0}&{\color{gray} 0}&{\color{gray} 0}\\
\bm{1}&\bm{1}&\bm{1}&\bm{1}&\bm{0}&\bm{1}&\bm{1}&\bm{1}&\bm{1}&\bm{1}&{\color{gray} 0}&{\color{gray} 0}&{\color{gray} 0}&{\color{gray} 0}&{\color{gray} 0}&{\color{gray} 0}\\
\bm{1}&\bm{1}&\bm{1}&\bm{1}&\bm{1}&\bm{0}&\bm{1}&\bm{1}&\bm{1}&\bm{1}&{\color{gray} 0}&{\color{gray} 0}&{\color{gray} 0}&{\color{gray} 0}&{\color{gray} 0}&{\color{gray} 0}\\
\bm{1}&\bm{1}&\bm{1}&\bm{1}&\bm{1}&\bm{1}&\bm{0}&\bm{1}&\bm{1}&\bm{1}&{\color{gray} 0}&{\color{gray} 0}&{\color{gray} 0}&{\color{gray} 0}&{\color{gray} 0}&{\color{gray} 0}\\
\bm{1}&\bm{1}&\bm{1}&\bm{1}&\bm{1}&\bm{1}&\bm{1}&\bm{0}&\bm{1}&\bm{1}&{\color{gray} 0}&{\color{gray} 0}&{\color{gray} 0}&{\color{gray} 0}&{\color{gray} 0}&{\color{gray} 0}\\
\bm{1}&\bm{1}&\bm{1}&\bm{1}&\bm{1}&\bm{1}&\bm{1}&\bm{1}&\bm{0}&\bm{1}&{\color{gray} 0}&{\color{gray} 0}&{\color{gray} 0}&{\color{gray} 0}&{\color{gray} 0}&{\color{gray} 0}\\
\bm{1}&\bm{1}&\bm{1}&\bm{1}&\bm{1}&\bm{1}&\bm{1}&\bm{1}&\bm{1}&\bm{0}&{\color{gray} 0}&{\color{gray} 0}&{\color{gray} 0}&{\color{gray} 0}&{\color{gray} 0}&{\color{gray} 0}\\
{\color{gray} 0}&{\color{gray} 0}&{\color{gray} 0}&{\color{gray} 0}&{\color{gray} 0}&{\color{gray} 0}&{\color{gray} 0}&{\color{gray} 0}&{\color{gray} 0}&{\color{gray} 0}&{\color{gray} 0}&{\color{gray} 0}&{\color{gray} 0}&{\color{gray} 0}&{\color{gray} 0}&{\color{gray} 0}\\
{\color{gray} 0}&{\color{gray} 0}&{\color{gray} 0}&{\color{gray} 0}&{\color{gray} 0}&{\color{gray} 0}&{\color{gray} 0}&{\color{gray} 0}&{\color{gray} 0}&{\color{gray} 0}&{\color{gray} 0}&{\color{gray} 0}&{\color{gray} 0}&{\color{gray} 0}&{\color{gray} 0}&{\color{gray} 0}\\
{\color{gray} 0}&{\color{gray} 0}&{\color{gray} 0}&{\color{gray} 0}&{\color{gray} 0}&{\color{gray} 0}&{\color{gray} 0}&{\color{gray} 0}&{\color{gray} 0}&{\color{gray} 0}&{\color{gray} 0}&{\color{gray} 0}&{\color{gray} 0}&{\color{gray} 0}&{\color{gray} 0}&{\color{gray} 0}\\
{\color{gray} 0}&{\color{gray} 0}&{\color{gray} 0}&{\color{gray} 0}&{\color{gray} 0}&{\color{gray} 0}&{\color{gray} 0}&{\color{gray} 0}&{\color{gray} 0}&{\color{gray} 0}&{\color{gray} 0}&{\color{gray} 0}&{\color{gray} 0}&{\color{gray} 0}&{\color{gray} 0}&{\color{gray} 0}\\
{\color{gray} 0}&{\color{gray} 0}&{\color{gray} 0}&{\color{gray} 0}&{\color{gray} 0}&{\color{gray} 0}&{\color{gray} 0}&{\color{gray} 0}&{\color{gray} 0}&{\color{gray} 0}&{\color{gray} 0}&{\color{gray} 0}&{\color{gray} 0}&{\color{gray} 0}&{\color{gray} 0}&{\color{gray} 0}\\
{\color{gray} 0}&{\color{gray} 0}&{\color{gray} 0}&{\color{gray} 0}&{\color{gray} 0}&{\color{gray} 0}&{\color{gray} 0}&{\color{gray} 0}&{\color{gray} 0}&{\color{gray} 0}&{\color{gray} 0}&{\color{gray} 0}&{\color{gray} 0}&{\color{gray} 0}&{\color{gray} 0}&{\color{gray} 0}
\end{array}
\right)
\]
\end{minipage}
}

\caption{Illustration of the indexing procedure. (a) An illustrative indexing map $\mathrm{id}:S' \to \{0,\dots,|S'|-1\}$ for a $4$-qubit system. (b) Adjacency matrix $H_M$ of the CTQW defined over the feasible subspace, where the feasible states are irregularly distributed (depicted in bold). (c) Adjacency matrix $\tilde{H}_M$ of the CTQW defined over the indexed subspace, which rearranges the feasible states into a contiguous block $\{0,\dots,|S'|-1\}$ (also shown in bold), enabling the walk to be implemented via the Fourier transform modulo $|S'|$.}
\label{fig:indexing}
\end{figure*}

For a constrained problem instance with a solution space $\mathcal S$ associated with a cost function $C$ and feasible solution space $\mathcal{S'} \subset \mathcal{S}$, QWOA starts by initializing the system in the uniform superposition over $\mathcal{S'}$, i.e., $\ket{\psi_0} = \frac{1}{\sqrt{|\mathcal{S'}|}}\sum_{z\in\mathcal{S'}}\ket{z}$, and then executes the layered variational circuit 
\begin{equation}
    U(\mathbf\gamma, \mathbf t) = \prod_{l=1}^p U_W(t_l) U_C(\gamma_l),
\end{equation}
in which
\begin{itemize}
    \item $p$ is the number of layers of QWOA, also referred to as the QWOA depth;
    \item $\mathbf{\gamma} = (\gamma_1, \gamma_2, \ldots, \gamma_p) \in \mathbb{R}^p$ and $\mathbf{t} = (t_1, t_2, \ldots, t_p) \in \mathbb{R}^p$ are the variational parameters or angles;
    \item $U_C(\gamma_l) =  e^{-i\gamma_l H_C}$, where $H_C$ --- the problem Hamiltonian --- acts as $H_C \ket{z} = C(z) \ket{z}$ for each solution $z \in \mathcal{S}$. In other words, $H_C$ is a diagonal operator with eigenvalues $C(z)$ for each solution $z \in \mathcal{S}$.
    \item $U_W(t_l) = e^{-it_l H_M}$, where $H_M$ is the mixing Hamiltonian, performs a CTQW on a graph with an adjacency matrix $H_M$ over the feasible subspace $S'$. Here, we consider a CTQW on a complete graph. Without loss of generality, we further take $H_M = J$, where $J$ is the all-ones matrix, since the resulting unitary operator is equivalent up to a global phase.
\end{itemize}
The operator $U_W(t_l)$ is implemented as
\begin{equation}
    U_W(t_l) = U^\dagger_{\#} e^{-it_l \tilde{H}_M} U_{\#}.
\end{equation}
Here, 
\begin{itemize}
    \item The indexing unitary $U_{\#}$ implement an efficiently computable bijective function $\operatorname{id}: \mathcal S' \rightarrow \{0, \cdots, |\mathcal S'|-1\}$ that lexicographically indexes the elements of the feasible solution space $\mathcal S'$, while the un-indexing unitary $U^\dagger_{\#}$ implements its inverse. The codomain of the function $\operatorname{id}$ is referred to as the indexed subspace.
    \item The operator $\tilde{H}_M$ is the mixing Hamiltonian corresponding to $H_M$ in the indexed subspace. 
\end{itemize}

In simple terms, the indexing and un-indexing operations are used to map the irregular feasible subspace $S'$ onto an ordered indexed subspace, where the CTQW can be implemented via the Fourier transform modulo $|\mathcal{S'}|$~\cite{cleve2000fast}. See Figure~\ref{fig:indexing} for an illustrative example. The indexing procedure can also be used to efficiently implement the initial state. In particular, the uniform superposition in the indexed basis, $\frac{1} {\sqrt{|\mathcal{S'}|}}\sum_{z=0}^{|\mathcal{S'}|-1}\ket{z}$,
can be prepared via the Fourier transform modulo $|\mathcal{S'}|$. Then, applying the un-indexing unitary $U_\#^{\dagger}$ yields the corresponding state over the feasible subspace,
$\ket{\psi_0} = \frac{1}{\sqrt{|\mathcal{S'}}|}\sum_{z\in\mathcal{S'}}\ket{z}$.

Efficient indexing and un-indexing algorithms are not available in general for arbitrary combinatorial optimization problems. Indeed, there exist problems, such as Minimal Vertex Cover and Maximal Cliques, for which even counting the number of solutions is computationally hard~\cite{valiant1979complexity}. The existence of such algorithms depends on the combinatorial organization of the feasible objects. As discussed by Marsh and Wang~\cite{marsh2020combinatorial}, several well-structured families admit such procedures, including $k$-combinations (e.g., the $k$-Densest Subgraph and Max Bisection problems), permutations (e.g., Traveling Salesman Problem), and Catalan-type lattice paths. Additional examples of combinatorial classes with efficient indexing can be found in Loehr’s book~\cite{bijective}. There also exist efficient indexing schemes for highly applied combinatorial optimization problems, such as Capacitated Vehicle Routing~\cite{bennett2021quantum1} and the Portfolio Optimization~\cite{slate2021quantum}. More broadly, many recursively defined combinatorial families fall under Wilf’s framework~\cite{wilf1977unified}, where efficient indexing follows from recursive counting rules.

The state $U(\mathbf{\gamma}, \mathbf{t}) \ket{\psi_0}$ --- called the QWOA final state ---  is prepared and measured in the computational basis. The measurement outcome is a bitstring from which the cost can be efficiently computed. This process is repeated multiple times in order to estimate the expectation value of the problem Hamiltonian $H_C$, i.e., $\bra{U(\mathbf \gamma, \mathbf t)} H_C \ket{U(\mathbf \gamma, \mathbf t)}$, which is optimized using some classical optimization algorithm, such as Nelder-Mead~\cite{nelder_mead}, by iteratively updating the parameters $(\mathbf{\gamma}, \mathbf{t})$. As shown by Marsh and Wang~\cite{marsh2019quantum}, an efficient execution of QWOA (and likewise QAOA) requires that the input optimization problem belongs to the $\mathsf{NPO}\text{-}\mathsf{PB}$ class, ensuring that the expectation value can be estimated to fixed precision using a polynomial number of circuit executions. Naturally, another necessary condition for the efficiency of the algorithm is the existence of efficient indexing and un-indexing procedures.

Note that QWOA can also be applied to unconstrained problems by omitting the indexing and un-indexing operators. Furthermore, the original QAOA is a special case of QWOA where the indexing and un-indexing operators are also omitted, and the mixing Hamiltonian ---- known in this context as the transverse field ---- is given by $H_M = \sum_{j=1}^n X_j$, where $X_j$ is the Pauli $X$ operator applied to qubit $j$ and $n$ denotes the number of qubits. It is worth mentioning that the QAOA mixing operator effectively implements a CTQW on a hypercube graph.

A key property of QWOA is that its expectation value is invariant under any permutation of the basis states~\cite{headley2023problem}. This symmetry has fundamental implications. One of them is that the performance of the algorithm depends only on how the costs are statistically distributed, meaning that QWOA is inherently unable to exploit any structural features of the problem instance, such as correlations between solutions. As a consequence, the expectation value of QWOA can be written solely in terms of the probability distribution associated with the solution spaces (see Ref.~\cite{headley2023problem}) --- in other words, the performance of QWOA depends only on the spectrum of the problem Hamiltonian. At the same time, this agnostic character of QWOA with respect to the structure of the problem fundamentally bounds its performance to one that is analogous to Grover’s algorithm~\cite{grover1996fast} for unstructured search. To discuss this limitation, we define the following performance metric for QWOA.

\begin{definition}[Optimal QWOA depth]
The optimal QWOA depth for a problem $\mathcal P$, denoted by $p^*_{\mathcal{P}}$, is the minimum number of layers required by QWOA to solve $\mathcal{P}$. We write $p^*$ to denote the optimal QWOA depth for any given problem.
\end{definition}

Bridi and Marquezino~\cite{bridi2024analytical} and Xie \emph{et al.}~\cite{xie2025performance} independently proved that the probability of measuring any basis state on QWOA is bounded by a quadratic growth with respect to the number of layers. Applying this result to the states that encode optimal solutions, we can conclude that the optimal QWOA depth is lower bounded by
\begin{equation}\label{eqn:optimal_depth}
 p^* = \Omega\left(\sqrt{\frac{|S'|}{|S_{\text{opt}}|}}\right),
\end{equation} 
where $|S_{\text{opt}}|$ denotes the number of optimal solutions. Note that Eq.~\eqref{eqn:optimal_depth} implies that QWOA is bounded by a quadratic speed-up over the random sampling, a classical algorithm defined as follows.
\begin{definition}[Random sampling algorithm] \label{def:random_sampling}
Random sampling is a classical algorithm that independently and uniformly samples the feasible solution space a given number of times and outputs the best solution found according to the cost function.
\end{definition}
In particular, the random sampling finds an optimal solution with high probability using $\Theta(|S'|/ |S_{\text{opt}}|)$ samples. Observe that the limitation imposed by Eq.~\eqref{eqn:optimal_depth} for combinatorial optimization is analogous to the bound of the unstructured search problem in quantum computing, for which Grover's algorithm is optimal~\cite{singleOptimalGrover, generalOptimalGrover}. Another manifestation of this quadratic Grover-style speed-up bound, provided by Bridi and Marquezino~\cite{bridi2024analytical}, is as follows: with $p$ layers, QWOA can only reach an expectation value corresponding to the top~$\Omega(1/p^2)$ fraction of the best solutions \footnote{For a statement of this result using random variables to model the solution space, see Corollary $4$ of Ref.~\cite{bridi2024analytical}.}, whereas the random sampling with $p$ samples would achieve in expectation a solution in the top~$\Theta(1/p)$. We further define another important performance metric for QWOA as follows.

\begin{definition}[$c$-approximate QWOA depth]
The $c$-approximate QWOA depth, for a constant $c$ and problem $\mathcal P$, denoted $p^c_{\mathcal{P}}$, is the minimum number of layers required by QWOA to achieve an approximation ratio $c$ for $\mathcal{P}$. We write $p^c$ to denote the corresponding quantity for any given problem.
\end{definition}

\subsection{Dynamical Lie Algebra (DLA)}\label{app:DLA}
Dynamical Lie algebras (DLAs)~\cite{larocca2022diagnosing, larocca2023theory, meyer2023exploiting, ragone2024lie, allcock2024dynamical} are a useful tool for characterizing the expressiveness of the PQCs. All operations on a quantum circuit acting on a $d$-dimensional state correspond to an element of the unitary group $\operatorname{U}(d)$, generated by the Lie algebra $\mathfrak{u}(d)$. DLA has been shown to characterize the unitaries generated by parameterized quantum circuits of the form
\begin{equation}\label{eqn:PQC}
    U(\boldsymbol{\theta}) = \prod_{l=1}^L U_l(\boldsymbol{\theta}_l), \quad U_l(\boldsymbol{\theta}_l) = \prod_{k=0}^K e^{-iH_k\theta_{lk}},
\end{equation}
where the circuit consists of $L$ layers, each with a set of tunable parameters $\boldsymbol{\theta}_l$ associated with unitaries generated by Hamiltonians $\{H_k\}$~\cite{d2021introduction,zeier2011symmetry}. The set of generators forms the foundation of the DLA.
\begin{definition}[Set of generators, adapted from Ref.~\cite{larocca2022diagnosing}, Definition~2]\label{def:generators}
    Given a parametrized quantum circuit of the form in Eq.~\ref{eqn:PQC}, we define the set of generators $\mathcal G = \{H_k\}_{k=0}^K$ as the set (of size $|\mathcal G| = K+1$) of the Hermitian operators that generate the unitaries in a single layer of $U(\boldsymbol{\theta})$ \footnote{For simplicity, we do not impose the usual traceless condition on the generators, as this choice is convenient for the nature of our arguments. The only difference is the presence of identity terms, which merely introduce a global phase.}. 
\end{definition}
The DLA is a subspace of the operator space, generated by the span of nested commutators of the ansatz's set of generators. 
\begin{definition}[Dynamical Lie algebra (DLA), adapted from Ref.~\cite{larocca2022diagnosing}, Definition~3]\label{def:DLA}
    Let $\mathcal G$ be a set of generators (see Definition~\ref{def:generators}). The Dynamical Lie Algebra (DLA) $\mathfrak g$ is the subalgebra of $\mathfrak{u}(d)$ spanned by the repeated nested commutators of the elements in $\mathcal G$, i.e., 
    \begin{equation}
        \mathfrak g = \langle iH_0, \cdots, iH_K\rangle_{\operatorname{Lie}}\subseteq \mathfrak{u}(d),
    \end{equation} 
    Here, $\mathfrak{u}(d)$ is the unitary Lie algebra of degree $d$, i.e., the Lie algebra formed by the set of $d \times d$ skew-Hermitian matrices.  $\langle S\rangle_{\operatorname{Lie}}$ denotes the Lie closure, i.e., the set obtained by repeatedly taking the nested commutators between the elements in $S$.

\end{definition}

We use $\mathfrak g_{\operatorname{QWOA}, \mathcal P}$ to denote the DLA corresponding to the QWOA ansatz for an input problem $\mathcal P$ and $\mathfrak g_{\operatorname{QWOA}}$ to indicate the DLA of the QWOA ansatz for any given problem. The Hamiltonians $H_C$ and $H_M$ form the set of generators for QWOA.

The DLA quantifies the extent to which the circuit can explore the full unitary group or a restricted subspace, making it a fundamental measure of expressivity. If the DLA spans the full Lie algebra $\mathfrak{u}(d)$ for a $d$-dimensional Hilbert space, the ansatz is considered maximally expressive. However, in many practical scenarios, the DLA forms a lower-dimensional subalgebra, imposing structure on the reachable unitary transformations. In this sense, the dimension of the DLA plays a key role in quantifying the expressivity of a PQC. In fact, as pointed out by Larocca \emph{et al.}~\cite{larocca2023theory}, the dimension of the DLA provides an upper bound on the critical number of circuit parameters required for a PQC to be deemed overparametrized, and the maximal rank attainable by the quantum Fisher information matrix (QFIM)~\cite{liu2020quantum} and Hessian matrices~\cite{fort2019emergent}. This provides us with a formal sufficiency condition of overparametrization.

\begin{definition}[Overparametrization, adapted from Ref.~\cite{larocca2023theory}]
\label{def:op_suff}
    A parametrized quantum circuit is deemed overparametrized if the number of circuit parameters exceeds the dimension of the dynamic Lie algebra.
\end{definition}

Note that in Ref.~\cite{larocca2023theory}, a PQC is defined to be overparametrized when the number of parameters is large enough that the QFIM attains its maximal achievable rank and remains saturated upon the introduction of additional parameters. Since the rank of the QFIM cannot exceed the number of circuit parameters, it follows from the definition that the number of parameters required for overparametrization must be at least as large as the maximal achievable rank of the QFIM. Combining this observation with the fact that the dimension of the DLA upper bounds the maximal rank attainable by the QFIM, we arrive at the sufficiency condition in Definition~\ref{def:op_suff}.

From a geometric perspective, the rank of the QFIM quantifies the number of independent directions in state space that can be accessed by tuning the circuit parameters. In the overparametrized regime, the circuit possesses enough parameters to explore all independent directions permitted by its underlying Lie algebra.

\section{Main Results}\label{sec:main_results}

\subsection{Bounding the DLA Dimension for QWOA}
We provide an upper bound on the DLA dimension for QWOA that scales quadratically in $m$. Recall that $m$ is the number of distinct cost values over the feasible solutions, or equivalently, the number of distinct eigenvalues of $H_C$ restricted to the feasible subspace. The key idea behind the proof is the following. Since QWOA is invariant under permutations of basis states, solutions with the same cost can be grouped in the problem Hamiltonian $H_C$. Due to the diagonal structure of the problem Hamiltonian $H_C$ and the uniform structure of the mixing Hamiltonian $H_M$, the nested commutators of the QWOA generators will share the same partition structure of $m^2$ blocks, where each block has constant entries. This means that the degree of freedom in the DLA space generated by the QWOA ansatz is bounded by $\mathcal{O}(m^2)$.

\begin{theorem}[General bound of DLA dimension for QWOA]\label{thm:DLA_QWOA}
For any instance of a combinatorial optimization problem, the DLA dimension of QWOA satisfies $\operatorname{dim}(\mathfrak g_{\operatorname{QWOA}}) \leq m^2 + 1$. In particular, $\operatorname{dim}(\mathfrak g_{\operatorname{QWOA}}) \leq m^2$ when, for each $1 \leq j \leq m$, the cost value $x_j$ either has multiplicity $1$ ($k_j = 1$) or is $0$ ($x_j = 0$).
\end{theorem}

\begin{proof}
We first consider the unconstrained case. The problem Hamiltonian $H_C$ is diagonal in the computational basis and can be rewritten in a block diagonal form as
\begin{align}
    H_C=\begin{bmatrix}
        x_1 I_{k_1}&\bar{0}&\cdots&\bar{0}\\
        \bar{0}&x_2I_{k_2}&\cdots&\bar{0}\\
        \vdots&\vdots&\ddots&\vdots\\
        \bar{0}&\bar{0}&\cdots&x_mI_{k_m}\\
    \end{bmatrix},
\end{align}
where $I_x$ is the identity matrix with the dimension $x$ and $\bar{y}$ denotes a matrix with all elements equal to $y$. We call such a partitioning of $H_C$ into blocks
a \emph{block partitioning pattern}. The generator $H_M$ can be written using the block partitioning pattern as
\begin{align}
    H_M=\begin{bmatrix}
        \bar{1}&\bar{1}&\cdots&\bar{1}\\
        \bar{1}&\bar{1}&\cdots&\bar{1}\\
        \vdots&\vdots&\ddots&\vdots\\
        \bar{1}&\bar{1}&\cdots&\bar{1}\\
    \end{bmatrix}.
\end{align}

Let $\Lambda$ be the set of skew-Hermitian matrices with the same block partitioning pattern as $H_C$ and $H_M$ such that all elements in any block are the same. In this case, note that $i H_M\in\Lambda$, but (in general) $i H_C\notin\Lambda$. From the definition of matrix multiplication and the fact that the skew-Hermitian operators are closed under the commutator, it follows that $[i H_C,A] \in \Lambda$ and $[A,B]\in\Lambda$ for $A,B\in\Lambda$. Consequently, all nested commutators are elements of $\Lambda$. Observe that $\Lambda$ is isomorphic to the unitary Lie algebra $\mathfrak u(m)$, therefore their size is $\vert \Lambda\vert = m^2$. Considering the worst case, where $i H_C\notin\Lambda$, we can conclude that $\operatorname{dim}(\mathfrak g_{\operatorname{QWOA}}) \leq m^2 + 1$. On the other hand, when $i H_C\in\Lambda$, we can tighten the bound to $\operatorname{dim}(\mathfrak g_{\operatorname{QWOA}}) \leq m^2$. This case occurs if and only if, for each $1 \leq j \leq m$, $k_j = 1$ or $x_j = 0$, finishing the unconstrained case. 

Now, consider the constrained problems case. Observe that we can express $U_{C}(\gamma_l)$ as $e^{-i\gamma_l H_C} = U^\dagger_{\#} e^{-i\gamma_l \tilde{H}_C} U_{\#}$, where $\tilde{H}_C$ acts similarly to $H_C$, but over the basis states of the indexed subspace. Similarly, note that $U_{\#}\ket{\psi_0}$ is a uniform superposition over the indexed subspace. Since $U_{\#}$ and $U^\dagger_{\#}$ define a change of basis, and the DLA is invariant under basis transformations, the analysis can be carried out in the indexed subspace using the indexed operators $\tilde{H}_C$ and $\tilde{H}_M$ . In this case, since our argument for the unconstrained case is based on the spectrum of $H_C$, which is the same as that of $\tilde{H}_C$, the result holds for the constrained case as well.
\end{proof}

Note that our upper bound is intuitive given the structural properties of QWOA. Since the ansatz is agnostic to the combinatorial structure of the problem and its dynamics are fully determined by the spectrum of the problem Hamiltonian, it follows naturally that the expressivity --- as captured by the DLA dimension --- should depend only on spectral properties. In particular, the fact that the DLA dimension grows at most quadratically with the number of distinct eigenvalues reflects the limited degrees of freedom available when the ansatz cannot distinguish between solutions beyond their cost values.

Furthermore, unlike previous DLA analyses in the literature, which typically rely on representations in terms of Pauli operators and sophisticated tools from group theory (for instance, Ref.~\cite{allcock2024dynamical}), our approach in the QWOA setting takes advantage of two key structural features of the algorithm: permutation invariance of basis states and the uniformity of the Hamiltonian mixing. These properties enable a much simpler and more intuitive analysis, allowing us to work directly with the spectra of the problem Hamiltonians rather than their full operator structure. As a result, the dimension of the DLA can be bounded using elementary linear algebra and block-structured matrix arguments.

\subsection{Overparametrization in QWOA for NPO-PB problem}

For $\mathsf{NPO}\text{-}\mathsf{PB}$ problems that admit efficient indexing and un-indexing algorithms, we derive complexity-theoretic conditions under which QWOA requires overparameterization, either to solve the problem or to achieve any fixed approximation ratio. Our approach combines implications of Theorem~\ref{thm:DLA_QWOA} with new insights into the performance bound of QWOA established in Refs.~\cite{bridi2024analytical,xie2025performance}. Our findings are general within the context of QWOA, holding for any target problem, provided that two necessary conditions for the efficient implementation of QWOA are met: (a) the problem belongs to the $\mathsf{NPO}\text{-}\mathsf{PB}$ class and (b) it admits efficient indexing and un-indexing procedures.

We begin by observing an immediate consequence of Theorem~\ref{thm:DLA_QWOA}. For $\mathsf{NPO}\text{-}\mathsf{PB}$ problems, whose window of possible costs --- and consequently $m$ --- is bounded by a polynomial function in the input size, we have that the dimension of the DLA is also polynomially bounded, as stated in the following corollary.

\begin{corollary}\label{cor:NPO-BP-DLA}
Let $\mathcal P$ be an $\mathsf{NPO}\text{-}\mathsf{PB}$ problem with input size $s$. Thus, $\operatorname{dim}(\mathfrak g_{\operatorname{QWOA}, \mathcal P}) =  \mathcal{O}(\operatorname{poly}(s))$. 
\end{corollary}

Now that we have bounded the DLA dimension of QWOA for $\mathsf{NPO}\text{-}\mathsf{PB}$ problems, polynomially, we turn our attention to the question of finding a solution to these problems. Two novel consequences of the Grover-style quadratic speed-up bound of QWOA established in Refs.~\cite{bridi2024analytical, xie2025performance} are stated in the theorem below.

\begin{theorem}\label{thm:QWOA_nonpoly}
Let $\mathcal{P}$ be an $\mathsf{NPO}$ problem with efficient indexing and un-indexing algorithms and with input size $s$. 
\begin{enumerate}
    \item If $\mathcal P \notin \mathsf{BPPO}$, then $p^*_{\mathcal P} \notin \mathcal{O}(\operatorname{poly}(s))$. \label{item1:thmQWOAsuperpoly}
    \item If $\mathcal P \notin \mathsf{BP}\text{-}\mathsf{APX}$, then $p^c_{\mathcal P} \notin \mathcal{O}(\operatorname{poly}(s))$ for any constant $c$. \label{item2:thmQWOAsuperpoly}
\end{enumerate}
\end{theorem}

\begin{proof}
\leavevmode
\begin{enumerate} 
    \item If QWOA can find optimal solutions to a problem $\mathcal{P}$ using a polynomial number of layers, i.e., $p^*_{\mathcal P} = \mathcal{O}(\operatorname{poly}(s))$, then its quadratic speed-up bound over random sampling implies that random sampling can also find optimal solutions with high probability using a polynomial number of samples. Since $\mathcal{P}$ is an $\mathsf{NPO}$ problem with efficient indexing and un-indexing algorithms, each sample can be computed in polynomial time. In this way, the random sampling procedure runs in polynomial time and therefore $\mathcal{P}$ belongs to \textsf{BPPO}. Item \ref{item1:thmQWOAsuperpoly} follows from the contraposition of the implication $p^*_{\mathcal P} = \mathcal{O}(\operatorname{poly}(s)) \Rightarrow \mathcal{P} \in \textsf{BPPO}$.
    \item An aforementioned result by Bridi and Marquezino~\cite{bridi2024analytical} establishes that QWOA exhibits a quadratic speed-up bound over random sampling with respect to the top fraction of the best solutions achieved by the algorithm. Consequently, since a fixed approximation ratio corresponds to a fixed cost, which in turn corresponds to a fixed top fraction of the best solutions, this quadratic bound implies that if QWOA attains an approximation ratio $c$ with a polynomial number of layers, then random sampling must also reach it with high probability using a polynomial number of samples. Therefore, with analogous arguments of those used to prove item~\ref{item1:thmQWOAsuperpoly}, we can conclude that $\exists c: p^c_{\mathcal P} = \mathcal{O}(\operatorname{poly}(s)) \Rightarrow \mathcal{P} \in \mathsf{BP}\text{-}\mathsf{APX}$ and item~\ref{item2:thmQWOAsuperpoly} follows from the contraposition of this implication. 
\end{enumerate}
\end{proof}

Observe that although in general $p^*_{\mathcal P} \notin \mathcal{O}(\mathrm{poly}(s))$ does not imply that $p^*_{\mathcal P}$ grows superpolynomially, i.e., $p^*_{\mathcal P} \in \omega(\mathrm{poly}(s))$, this implication does hold when $p^*_{\mathcal P}$ is monotonically increasing. For most problems of interest, larger instances typically require deeper ansätze, making it reasonable to assume that the optimal QWOA depth is a monotonic function --- or at least asymptotically increasing. Under this assumption, non-polynomially bounded growth is effectively superpolynomial. Nevertheless, pathological cases may arise in which the optimal QWOA depth exhibits irregular behavior. For instance, it is theoretically possible that $p^*_{\mathcal P}$ grows polynomially for odd values of $s$ and superpolynomially for even values (and vice versa). In such cases, the function would not belong to $\mathcal{O}(\mathrm{poly}(s))$, yet would also fail to lie in $\omega(\mathrm{poly}(s))$.

We now present two remarks regarding item~\ref{item1:thmQWOAsuperpoly} of Theorem~\ref{thm:QWOA_nonpoly}. While these remarks are stated in the context of item~\ref{item1:thmQWOAsuperpoly}, analogous reasoning applies to item~\ref{item2:thmQWOAsuperpoly} by replacing $\mathsf{BPPO}$ and $\mathsf{PO}$ with $\mathsf{BP\text{-}APX}$ and $\mathsf{APX}$, respectively.
Firstly, the inverse of item~\ref{item1:thmQWOAsuperpoly} of Theorem~\ref{thm:QWOA_nonpoly} does not hold, i.e, $\mathcal{P} \in \textsf{BPPO}$ does not necessarily imply that the optimal QWOA depth is polynomially bounded in the input size. Indeed, the fact that a problem belongs to $\textsf{BPPO}$ does not guarantee that it can be efficiently solved by random sampling, which is the baseline performance that QWOA quadratically improves upon. Secondly, restricting the theorem to $\mathsf{PO}$ rather than $\mathsf{BPPO}$ would yield a stronger result. However, unless $\mathsf{PO} = \mathsf{BPPO}$, the arguments used in the proof of item~\ref{item1:thmQWOAsuperpoly} do not allow us to conclude that $\mathcal P \notin \mathsf{PO}$ implies $p^*_{\mathcal P} \notin \mathcal{O}(\operatorname{poly}(s))$. Whether there exists a problem in $\textsf{BPP} - \textsf{PO}$ that QWOA can still solve with polynomial depth remains an open question.

Now, combining Corollary~\ref{cor:NPO-BP-DLA} and Theorem~\ref{thm:QWOA_nonpoly} along with Definition~\ref{def:op_suff}, we can derive the following theorem, the promised result of this section.

\begin{theorem}[Conditions for QWOA Overparameterization]\label{thm:NPOPB} Let $\mathcal{P}$ be an $\mathsf{NPO}\text{-}\mathsf{PB}$ problem with an efficient indexing and algorithms. 
\begin{enumerate}
    \item If $\mathcal P \notin \textsf{BPPO}$, then QWOA requires overparameterization to solve $\mathcal P$. \label{item1:thmNPOPB}
    \item If $\mathcal P \notin \mathsf{BP}\text{-}\mathsf{APX}$, then QWOA requires overparameterization to achieve any fixed approximation ratio for $\mathcal P$. \label{item2:thmNPOPB}
\end{enumerate}
\end{theorem}

\section{Illustrative Examples} \label{sec:examples}

To illustrate the implications of our theoretical results (related to Theorem~\ref{thm:DLA_QWOA} and Corollary~\ref{cor:NPO-BP-DLA}), we now consider concrete instances of optimization problems --- specifically, unstructured search, Max-Cut, and $k$-Densest Subgraph. Note that these problems do not satisfy the premises of Theorems~\ref{thm:QWOA_nonpoly} and~\ref{thm:NPOPB}, since these results are conditional statements whose assumptions rely on classical complexity-theoretic separations --- for example, the existence of an $\mathsf{NPO}$ problem outside $\mathsf{BPPO}$. Since separations of this kind are not known, these theorems cannot be directly instantiated by standard benchmark problems. Indeed, proving the existence of such a problem would constitute a major advance in computational complexity, with implications that could reverberate in fundamental questions of complexity theory such as $\mathsf{P}$ versus $\mathsf{NP}$. The purpose of Theorems~\ref{thm:QWOA_nonpoly} and~\ref{thm:NPOPB} is therefore primarily conceptual, as they clarify the consequences for optimization problems that are not efficiently solvable by a randomized algorithm.

\subsection{Unstructured Search}

Consider a black box for the cost function $C: \mathcal{S} \rightarrow \{0, 1\}$. The inputs $z\in M$, where $M = \{z\in \mathcal{S}: C(z) = 1\}$ are known as the marked elements. In the unstructured search problem, $M$ is assumed to be nonempty, and the goal is to find a marked element~\cite{grover1996fast}. For simplicity, we assume that this problem is unconstrained with $|\mathcal{S}| = |\mathcal{S'}| = 2^n$, where $n$ is the number of qubits. Although this is typically treated as a search-type problem, the unstructured search can also be formulated as an optimization problem --- for example, see Ref.~\cite{jiang2017near}. 

The unstructured search problem is the simplest non-trivial instance that illustrates Theorem~\ref{thm:DLA_QWOA}. In this particular case, $m = 2$ and the two distinct eigenvalues represent marked or unmarked solutions. This yields the following corollary.
\begin{corollary}\label{cor:QWOA_search}
For any instance of the unstructured search problem, the DLA dimension of QWOA satisfies $\operatorname{dim}(\mathfrak g_{\operatorname{QWOA, Search}})\leq 5$. 
\end{corollary}
Let $R$ be the ratio of marked to total solutions. From the lower bound on the unstructured search problem~~\cite{bennett1997strengths, boyer1998tight, singleOptimalGrover, generalOptimalGrover}, we must have $p^*_{\operatorname{Search}} = \Omega(1/\sqrt{R})$. For typical problems, $p^*_{\operatorname{Search}}$ must scale as a superconstant with respect to the number of qubits and for certain cases may even scale exponentially. This suggests that a highly overparameterized ansatz is required to solve this problem. The exact DLA dimension for the unstructured search problem is obtained in Theorem~\ref{thm:QWOA_search_exact}.

\begin{theorem}\label{thm:QWOA_search_exact}
For the unstructured search problem, the DLA dimension of QWOA is given by   
\begin{equation}
    \operatorname{dim}(\mathfrak{g}_{\operatorname{QWOA, Search}}) = 
    \begin{cases}
        4, & |M| = 1 \\ 
        5, & \text{otherwise}.
    \end{cases}
\end{equation}
\end{theorem}

\begin{proof}
    Let
    \begin{equation}
        H_1=iH_C = i\begin{pmatrix}
            I & \overline{0}\\ \overline{0} & \overline{0}\\
        \end{pmatrix}, \quad 
        H_2=i H_M = i\begin{pmatrix}
            \overline{1} & \overline{1}\\ \overline{1} & \overline{1}\\ 
        \end{pmatrix}. 
    \end{equation}  
We calculate the nested commutators $H_3 = [H_1, H_2]$, $H_4 = [H_1, H_3]$, and $H_5 = [H_2, H_3]$, obtaining
\begin{align}
        H_3 = -\begin{pmatrix}
            \overline{0} & \overline{1}\\ \overline{-1} & \overline{0}\\
        \end{pmatrix}, \quad 
        H_4 = -i\begin{pmatrix}
            \overline{0} & \overline{1}\\ \overline{1} & \overline{0} \end{pmatrix}, \\\nonumber
            H_5 = -i \begin{pmatrix}
            \overline{2(|M| - |S)} & \overline{2|M| - |S|}\\ \overline{2|M| - |S|} & \overline{2|M|}\\ 
        \end{pmatrix}.
    \end{align}
If $|M| > 1$, one can show that $\{ H_1, H_2, H_3, H_4, H_5 \}$ is a linearly independent set, reaching the tight bound of Theorem~\ref{thm:DLA_QWOA} with $\operatorname{dim}(\mathfrak{g}_{\operatorname{QWOA, Search}}) = 5$. Otherwise, $|M| = 1$ and Theorem~\ref{thm:DLA_QWOA} implies that $\operatorname{dim}(\mathfrak{g}_{\operatorname{QWOA, Search}}) \leq 4$. Here, the tight bound is achieved because $H_1$ is a block-constant matrix that belongs to the space spanned by $\{ H_2, H_3, H_4, H_5 \}$.
\end{proof}

\subsection{Max-Cut}\label{sec:max-cut}
Consider an undirected, unweighted graph $G=(V,E)$ without loops, where $V = \{1,\dots,n\}$ is the set of vertices and $E \subseteq \{(u,v) : \ u,v\in V, u\neq v\}$ is the set of edges. Max-Cut is an unconstrained problem that seeks to partition the vertex set $V$ into two subsets, $T$ and $V\backslash T$, such that the number of edges with one endpoint at $T$ and the other at $V\backslash T$ is maximized~\cite{blekos2024review}. These edges are referred to as the edges crossing the cut or cut edges. Equivalently, the objective is to find a bipartite subgraph of $G$ that contains the largest possible number of edges. This problem is unconstrained and $|\mathcal{S}| = |\mathcal{S'}| = 2^n$. Classically, we encode the solution to the Max-Cut problem with a bitstring $z=z_1\dots z_n$, where each $z_i$ is either $0$ or $1$, indicating the subset to which vertex $i$ belongs. In the corresponding quantum embedding, this bitstring maps to a computational basis state of an $n$-qubit system $\ket{z} = \ket{z_1\dots z_n}$.

Before we apply our theoretical results to the Max-Cut problem, we first show two claims regarding the number of optimal solutions for $2$-regular, chain, and complete graphs for this problem. 
\begin{claim} \label{claim:maxcut1}
    For $2$-regular and chain graphs on the Max-Cut problem, the number of optimal solutions is bounded by $|S_{\text{opt}}| = \mathcal{O}(n)$. 
\end{claim}

\begin{proof}
For even $2$-regular graphs and general chains, we have $|S_{\text{opt}}| = 2$, since both are bipartite graphs and the optimal solution is the unique bipartition, up to bitwise complement.. In contrast, for odd $2$-regular graphs, we have $|S_{\text{opt}}| = 2n$. This is due to the observation that the optimal solutions alternate bits between adjacent vertices as much as possible, with exactly one pair of adjacent vertices having equal bits. There are $n$ such configurations, and combining them with their respective bitwise complements yields a total of $2n$ optimal solutions. Therefore, number of optimal solutions satisfy $|S_{\text{opt}}| = \mathcal{O}(n)$ for both $2$-regular and chain graphs.
\end{proof}

\begin{claim} \label{claim:maxcut2}
    For the complete graph on the Max-Cut problem, the number of optimal solutions is of order $|S_{\text{opt}}| = \Theta(2^n / \sqrt{n})$. 
\end{claim}

\begin{proof}
Note that the Max-Cut cost for complete graphs depends solely on the sizes of the subsets of the partition. In particular, suppose that we have $j$ vertices in one part of the partition. Then, the Max-Cut cost is given by $j(n - j)$. The values of $j$ that maximize the cost function are $j = n/2$ if $n$ is even, and $j = (n-1)/2, (n + 1)/2$ if $n$ is odd. The number of optimal solution in even case is $\binom{n}{n/2}$, while in odd case it is $\binom{n}{(n-1)/2} + \binom{n}{(n+1)/2}$. In both cases, $|S_{\text{opt}}| = \Theta(\binom{n}{n/2})$. To finish, with Stirling's approximation, we can conclude that $|S_{\text{opt}}| = \Theta(2^n / \sqrt{n})$.
\end{proof}

Now, note that Max-Cut is a $\mathsf{NPO}\text{-}\mathsf{PB}$ problem. Therefore, by Corollary~\ref{cor:NPO-BP-DLA}, its DLA dimension is polynomial in the input size. To be more precise, observe that the number of possible edges of $G$ is bounded by $\mathcal{O}(n^2)$. Since Max-Cut is a problem where we count edges, then $m = \mathcal{O}(n^2)$, leading us --- from Theorem~\ref{thm:DLA_QWOA} --- to the following corollary.
\begin{corollary}
For any instance of the Max-Cut problem, the DLA dimension of QWOA satisfies  $\operatorname{dim}(\mathfrak g_{\operatorname{QWOA, Max-Cut}}) = \mathcal{O}(n^4)$. 
\end{corollary}
The bound can be made tighter for Max-Cut on special classes of graphs, for example, on the $2$-regular and chain graphs. In both cases, as the number of edges is of the order of $n$, the number of distinct feasible costs satisfies $m = \mathcal{O}(n)$, and therefore the dimension of the DLA is bounded by $\mathcal{O}(n^2)$. Now, we turn our attention to the optimal QWOA depth to solve Max-Cut on these instances, given by Eq.~\eqref{eqn:optimal_depth}. For Max-Cut, the size of the feasible solution space is $|S'| = 2^n$. In addition, as shown in Claim~\ref{claim:maxcut1}, the number of optimal solutions for these particular instances satisfies $|S_{\text{opt}}| = \mathcal{O}(n)$, implying that $p^*_{\operatorname{Max-Cut}} = \Omega(\sqrt{2^n/n})$. Coupling these bounds on optimal QWOA depth with the bound on the DLA dimension, we observe that the exponential gap between both quantities implies a highly overparameterized QWOA ansatz to successfully solve Max-Cut on these instances. Observe that this is in contrast with the aforementioned QAOA results, where the Max-Cut problem on $2$-regular graphs and chains can be solved with a number of layers of the same order as the DLA dimension~\cite{larocca2023theory}.

In the case of the complete graph, the situation changes drastically. Similarly to previous cases, the DLA dimension is bounded by $\mathcal{O}(n^2)$, which follows from the fact that the Max-Cut cost depends solely on the sizes of the subsets of the partition, leading to at most $\mathcal{O}(n)$ distinct costs. However, in stark contrast to the exponential depth required for $2$-regular and chain graphs, the optimal QWOA depth is polynomially bounded by $p^*_{\operatorname{Max-Cut}} = \Omega(n^{1/4})$. This follows from the fact proved in Claim~\ref{claim:maxcut2}, which shows that the number of optimal solutions is of order $|S_{\text{opt}}| = \Theta(2^n / \sqrt{n})$. Observe that Max-Cut on the complete graph is in $\mathsf{BPPO}$ class. If our bound for the DLA dimension is at least asymptotically close to the exact dimension and the QWOA indeed achieves a quadratic speed-up over the random sampling, as corroborated by the numerical evidence of Refs.~\cite{zhang2024grover, xie2025performance}, QWOA can solve these instances with an underparametrized ansatz.

\subsection{\texorpdfstring{$k$}{}-Densest Subgraph}
Consider a graph $G = (V, E)$ exactly as in the Max-Cut problem and an integer $k$ such that $1 < k < n$. The $k$-Densest Subgraph is a constrained problem that seeks the subgraph $G'$ with exactly $k$ vertices that contains the most edges~\cite{golden2023numerical}. Similarly to Max-Cut, we encode each solution in a computational basis state of an $n$-qubit system, where the corresponding bitstring $z = z_1 \dots z_n$ indicates the selected vertices: $z_i = 1$ if vertex $i$ is included in the subgraph $G'$, and $z_i = 0$ otherwise. On the other hand, unlike Max-Cut, $k$-Densest Subgraph is a constrained problem. In this case, feasible solutions are restricted to bitstrings with Hamming weight exactly $k$. Therefore, $|\mathcal{S}| = 2^n$ and $|\mathcal{S'}| = \binom{n}{k}$. Efficient indexing and un-indexing operators for the set of $k$-combinations are provided by Marsh and Wang~\cite{marsh2020combinatorial}.

Similarly to Section~\ref{sec:max-cut}, we require the following claim on the number of optimal solutions before we proceed with the application of our theoretical results.
\begin{claim} \label{claim:densest}
    For $2$-regular and chain graphs on the $k$-Densest Subgraph problem, the number of optimal solutions is bounded by $|S_{\text{opt}}| = \mathcal{O}(n)$. 
\end{claim}

\begin{proof}
For both $2$-regular and chain graphs, observe that the cost of a solution is given by $k - x$, where $x$ is the number of connected components of the subgraph $G'$. Therefore, optimal subgraphs have only one connected component, that is, they are composed of consecutive vertices of the cycle ($2$-regular) or the path (chain). In the $2$-regular graph, there are exactly $n$ such optimal subgraphs, while in the chain graph, there are $n + 1 - k$ possibilities. Therefore, it follows that $|S_{\text{opt}}| = \mathcal{O}(n)$ for both $2$-regular and chain graphs.
\end{proof}

As Max-Cut, the $k$-Densest Subgraph is a $\mathsf{NPO}\text{-}\mathsf{PB}$ problem that counts edges, leading to the corollary below.
\begin{corollary}
For any instance of the $k$-Densest Subgraph problem, the DLA dimension of QWOA satisfies $\operatorname{dim}(\mathfrak g_{\operatorname{QWOA, k-Densest-Subgraph}}) = \mathcal{O}(n^4)$. 
\end{corollary}
For the same classes of graphs considered in Max-Cut, the bound is actually tighter. The complete graph is a trivial instance since any subgraph of the complete graph is a clique. Therefore, $m = 1$ and so we can ignore it. For $2$-regular and chain graphs, we can conclude that $m = \mathcal{O}(n^2)$ with analogous arguments to the Max-Cut problem. Now, we consider the optimal QWOA depth. We prove in Claim~\ref{claim:densest} that the number of optimal solutions for both graphs is bounded by $\mathcal{O}(n)$. It remains to consider the size of the feasible solution space. Firstly, observe that a given class of instances belongs to $\mathsf{PO}$ if $k$ is a constant, as it could be efficiently solved by brute force. This follows from the following facts: the number of feasible solutions is the polynomial $\binom{n}{k} = \mathcal{O}(n^k)$; the feasible space can be efficiently indexed; and the problem is in $\mathsf{NPO}$, so each solution can be verified in polynomial time. Then, we now assume that $k$ is not a constant, but growing as a function, i.e., $k = k(n)$. In this setting, the optimal QWOA depth is bounded by $p^*_{k\operatorname{-Densest-Subgraph}} = \Omega\left(\binom{n}{k(n)}/n\right)$. It is reasonable to assume that for the most challenging instances, $\Omega\left(\binom{n}{k(n)}/n\right) \neq \mathcal{O}(\operatorname{poly}(n))$, since otherwise we can conclude analogously that the instances also belong to $\mathsf{PO}$ class. In this case, we would require an overparameterized ansatz to solve the $k$-Densest Subgraph problem.

\section{Conclusion and Outlook}\label{sec:conclusion}
In this work, we have derived a general upper bound on the DLA dimension of QWOA. In particular, we showed that this bound scales polynomially in the number of distinct cost values of the cost function. An implication of our result is that, for optimization problems in the class $\mathsf{NPO}\text{-}\mathsf{PB}$, the DLA dimension is polynomially bounded in the input size. This structural insight allows us to derive meaningful consequences for the expressivity of QWOA. By connecting our bound to recent analytical results on QWOA’s performance limitations, we identified complexity-theoretic conditions under which the QWOA ansatz must be overparameterized to either achieve exact optimality or reach a fixed approximation ratio. In particular, for problems outside the $\mathsf{BPPO}$ and $\mathsf{BP\text{-}APX}$ classes, overparameterization --- often of exponential magnitude --- becomes necessary. This stands in sharp contrast with QAOA, where numerical evidence indicates that optimal solutions can be obtained with high probability as the ansatz becomes only slightly overparameterized~\cite{larocca2023theory}. Our result, therefore, suggests that the DLA dimension cannot always be taken as a reliable indicator of convergence in VQAs.

Several open questions remain. One important path for future work is to derive non-trivial lower bounds on the DLA dimension of QWOA, which would help determine the tightness of our results and offer a more complete characterization of the algorithm’s expressivity. It is also worth exploring how expressivity, as quantified by DLA dimension, relates to generalization behavior in noisy or data-driven contexts, particularly in hybrid quantum-classical settings. Another promising avenue concerns the classical simulability of QWOA. Specifically, the results of Goh \emph{et al.}~\cite{goh2023lie}, combined with our Corollary~\ref{cor:NPO-BP-DLA}, suggest that the expectation value of QWOA could be computed in polynomial time for $\mathsf{NPO}\text{-}\mathsf{PB}$ problems, provided certain prerequisites are met, such as the knowledge of a basis for the DLA. This potential could lead to classical simulation methods competitive with the approach of Golden \emph{et al.}~\cite{GM-Th-QAOA}. Furthermore, our Lie-algebraic framework could be applied more broadly to other variational quantum algorithms, such as QAOA and adaptive ansätze, potentially leading to unified expressivity measures across variational paradigms. We hope that the simplicity and generality of our approach provide a foundation for these future directions and contribute to a deeper theoretical understanding of quantum optimization algorithms.

Another promising direction is how our result can inspire the development of new ansätze. In this sense, QWOA is structurally agnostic by design, which grants generality but limits its performance to Grover-style quadratic speed-up. Motivated by the limitations this structural agnosticism imposes on mixing design, a growing trend in the literature is to investigate how incorporating problem structure into quantum algorithms can help overcome such constraints. Notable works in this direction include: Headley~\cite{headley_thesis} used a statistical approach to analyze QAOA with both the original transverse field mixing and the line mixing, showing how the algorithm can capture dependencies between solutions in both settings; Matwiejew and Wang~\cite{matwiejew2024quantum} introduced a variant of QWOA called Quantum Walk Optimization Algorithms on Combinatorial Subsets (QWOA-CS), where the continuous-time quantum walk incorporates problem structure via an algebraic framework based on graph automorphisms; and more broadly, Matwiejew, Pye, and Wang~\cite{matwiejew2023quantum} consider structural exploitation beyond combinatorial optimization problems, extending the discussion to continuous optimization settings. The results and insights presented here suggest a potential impact on this growing line of research, which warrants further investigation. In particular, exploring how our framework can inform the design of structure-aware mixing may help uncover mechanisms that enable performance surpassing that of Grover’s algorithm. In this context of tailored ansätze for specific problem families, we can explore alternative approaches that incorporate problem structure, such as those incorporating nonlocal operators or time-dependent dynamics.

\begin{acknowledgments}  
The authors would like to thank Andris Ambainis for valuable input on computational complexity classes. 
G.A.B. thanks the financial support of Coordena\c{c}\~{a}o de Aperfei\c{c}oamento de Pessoal de N\'{i}vel Superior (CAPES/Brazil) Grants No. 88887.950125/2024-00 and 88881.128522/2025-01. 
D.L. gratefully acknowledges funding from QuantERA Project QOPT. 
L.P. acknowledges support by the National Research Foundation, Singapore, and A*STAR under its CQT Bridging Grant and its Quantum Engineering Programme under grant NRF2021-QEP2-02-P05. 
% L.P. LP: CQT statement above covers my funding support.
\sloppy
R.A.M.S. thanks the financial support from the Latvian Quantum Initiative under the European Union Recovery and Resilience Facility project number 2.3.1.1.i.0/1/22/I/CFLA/001.
F.L.M. thanks the financial support of Conselho Nacional de Desenvolvimento Científico e Tecnol\'{o}gico (CNPq/Brazil) Grants No. 407296/2021-2 and 306049/2025-2.
S.A. thanks the financial support from FY23 SUG-PATRICK REBENTROST (WBS Code: A00098700000).
\end{acknowledgments}

\bibliographystyle{apsrev4-2}
\bibliography{main}% Produces the bibliography via BibTeX.

%\nocite{*}

\end{document}